\documentclass[11pt,letterpaper]{article}
\usepackage[margin=1in]{geometry}
\usepackage{soul}
\usepackage{url}
\usepackage[hidelinks]{hyperref}
\usepackage[utf8]{inputenc}
\usepackage[small]{caption}
\usepackage{graphicx}
\usepackage{xcolor}
\usepackage{amsmath}
\usepackage{amsthm}
\usepackage{amsfonts}
\usepackage{booktabs}
\usepackage{algorithm}
\usepackage{algorithmic}
\usepackage[switch]{lineno}
\usepackage{thm-restate}

\newtheorem{theorem}{Theorem}
\newtheorem{proposition}{Proposition}

\newtheorem{lemma}{Lemma}

\DeclareMathOperator*{\adisc}{\mathrm{asymdisc}}
\DeclareMathOperator*{\wdisc}{\mathrm{wdisc}}

\title{Tight Asymptotic Bounds for Fair Division With Externalities}
\author{
\textbf{Frank Connor}\\
McGill University\\
\texttt{frank.connor@mail.mcgill.ca}
\and
\textbf{Max Dupr\'{e} la Tour}\\
McGill University\\
\texttt{maxduprelatour@gmail.com}
\and
\textbf{Vishnu V. Narayan}\\
Indian Institute of Technology Bombay\\
\texttt{vishnu.narayan@mail.mcgill.ca}
\and
\textbf{Šimon Schierreich}\\
AGH University of Krakow,\\
Czech Technical University in Prague\\
\texttt{schiesim@fit.cvut.cz}
}

\date{January 19, 2026}

\usepackage{xspace}
\usepackage{mathtools}
\usepackage{cleveref}
\usepackage{csquotes}
\usepackage{natbib}
\newcommand{\EFa}{EF$_{a}$\xspace}
\newcommand{\EFac}{EF$_a$-$c$\xspace}

\newcommand{\Oh}[1]{{\mathcal{O}\left(#1\right)}}

\newcommand{\Omh}[1]{\Omega\left(#1\right)}
\newcommand{\hy}{\hbox{-}\nobreak\hskip0pt}
\NewDocumentCommand{\cc}{ O{} O{} m }{\mbox{%
    \expandafter\ifx\expandafter\relax\detokenize{#2}\relax\else{#2\hy}\fi%
    \textsf{#3}%
    \expandafter\ifx\expandafter\relax\detokenize{#1}\relax\else{\hy#1}\fi%
    }\xspace}

\let\oldabstract\abstract
\let\oldendabstract\endabstract
\makeatletter
\renewenvironment{abstract}
{%
               {\list{}{\addtolength{\leftmargin}{1.1em} 
                        \listparindent 1.5em%
                        \itemindent    \listparindent%
                        \rightmargin   \leftmargin%
                        \parsep        \z@ \@plus\p@}%
                \item\relax}%
               {\endlist}%
\oldabstract}
{\oldendabstract}
\makeatother

\begin{document}

\maketitle

\begin{abstract}
    We study the problem of allocating a set of indivisible items among agents whose preferences include externalities. Unlike the standard fair division model, agents may derive positive or negative utility not only from items allocated directly to them, but also from items allocated to other agents. Since exact envy-freeness cannot be guaranteed, prior work has focused on its relaxations. However, two central questions remained open: does there always exist an allocation that is envy-free up to one item (EF1), and if not, what is the optimal relaxation EF-$k$ that can always be attained?
    
    We settle both questions by deriving tight asymptotic bounds on the number of items sufficient to eliminate envy. We show that for any instance with $n$ agents, an allocation that is envy-free up to $\Oh{\sqrt{n}}$ items always exists and can be found in polynomial time, and we prove a matching $\Omh{\sqrt{n}}$ lower bound showing that this result is tight even for binary valuations, which rules out the existence of EF1 allocations when agents have externalities.
\end{abstract}


\section{Introduction}

Fair division of indivisible items, including administrative tasks, laboratory equipment, or heritage between heirs, is an active research area that connects computer science and economics~\citep{BramsT1996,Thomson2016,AmanatidisABFLMVW2023,NguyenR2023}.
Traditionally, a fair division instance consists of a set of indivisible items and a set of agents, where each agent has a valuation function assigning a numerical utility to each possible bundle. Crucially, in the standard model, the utility of an agent depends solely on the items allocated to it. The goal is to partition the items among agents so that a given fairness criterion is satisfied. A prominent example is \emph{envy-freeness} (EF)~\citep{Foley1967}, which requires that each agent weakly prefers its own bundle over the bundle of any other agent. Since EF allocations are not guaranteed to exist (consider two agents and one positively-valued item), several relaxations have been studied, most notably \emph{envy-freeness up to one item} (EF1)~\citep{LiptonMMS2004,Budish2011,CaragiannisKMPS19} and \emph{envy-freeness up to any item}~(EFX)~\citep{CaragiannisKMPS19,PlautR2020}. It is known that EF1 allocations always exist in the standard fair division model~\mbox{\citep{LiptonMMS2004}}.

Nevertheless, there are many scenarios where the utility an agent experiences cannot depend solely on its own bundle. Consider a dormitory allocating newly purchased equipment among students: a student benefits not only from items given directly to them, but also from items assigned to their roommate, with whom they can share. Conversely, when dividing resources among competing firms, such as exclusive licenses, patents, or territory rights, a firm may be harmed when a rival acquires a valuable asset, even if the firm's own allocation remains unchanged. Externalities can also be mixed: when distributing research grants among academics, a researcher may benefit if a collaborator receives funding for a joint project but suffer if a competitor in the same field does.

All these examples have one aspect in common: the inherent \emph{externalities} underlying the agents' valuations. Although externalities are well studied in the broader economics literature on resource allocation~\citep{AyresK1969,KatzS1985,Velez2016}, most prior work in fair division focuses on \emph{divisible} items~\citep{BranzeiPZ2013,LiZZ2015}. The model of fair division of indivisible items with externalities was formulated only recently~\citep{SeddighinSG2019,MishraPG2022,AzizSSW2023,DeligkasEKS2024}. Most relevant to our work are \citet{AzizSSW2023}, who adapted the standard envy-based fairness notions to the externalities setting and proved the existence of EF1 allocations for some special cases, and \citet{DeligkasEKS2024}, who investigated the computational complexity of deciding whether a fair allocation exists. Despite this progress, perhaps the most fundamental question has remained open:
\begin{center}
    Is an EF1 allocation guaranteed to exist for every instance of fair division with externalities?
\end{center}
Our current understanding of this question is limited. EF1 allocations are known to always exist for two agents, and for three agents with binary valuations and no chores~\citep{AzizSSW2023}.\footnote{See \Cref{sec:prelims} for a formal definition of chores in the externalities model.} Beyond these special cases, even the existence of EF-$k$ allocations for any constant $k$ was unknown. \citet{AzizSSW2023} explicitly state that ``an important open question that remains from our work is whether there always exists an EF1 allocation [...] if the answer is negative, it would be reasonable to ask for the optimal relaxation EF-$k$ that can be attained.'' \citet{DeligkasEKS2024} conjectured that EF1 allocations always exist for binary valuations, acknowledging that a proof ``seems highly non-trivial'' given the extensive case analysis and computer verification required even for three agents.


\subsection{Our Contribution}

We show that EF1 allocations do not always exist, even if agents' preferences are binary, thus resolving the open question posed by both \citet{AzizSSW2023} and \citet{DeligkasEKS2024}. We supplement this by showing that the optimal relaxation EF-$k$ that can be attained is for $k=\Theta(\sqrt{n})$, where~$n$ is the number of agents, answering the second open question of \citet{AzizSSW2023}.

After a formal definition of the model of fair division with externalities in \Cref{sec:prelims}, we introduce the technical tools that we use to prove our results in \Cref{sec:toolbox}. There are two main ingredients. First, we introduce an \emph{asymmetric envy model} and show its equivalence with the externalities model. Valuations in the asymmetric envy model are closer to the standard additive valuations than the valuations in the original externalities model. That said, this model is mostly of theoretical interest: the valuations are not easy to interpret, and even a social-welfare maximizing allocation is not well-defined with these valuations. The second key component of our work is the framework of \emph{multicolored discrepancy}~\citep{DoerrS2003,ManurangsiS2022,CaragiannisLS2025,ManurangsiM2026}. Importantly, we generalize some of the discrepancy bounds known from prior work to settings with a mixture of goods and chores, which may be of independent interest.

In \Cref{sec:upperbound}, we use the technical toolbox developed in \Cref{sec:toolbox} to show that, for any fair division with externalities, there exists a deterministic polynomial-time mechanism~$\mathcal{M}_1$ that finds an allocation that is envy-free up to $\Oh{\sqrt{n}}$ items. Importantly, in \Cref{sec:lowerbound}, we show that the mechanism is asymptotically optimal. Specifically, we show that there is an instance in the externalities model such that the best relaxation EF-$k$ we can achieve is for $k\in\Omh{\sqrt{n}}$, even if the valuations are binary and there are no chores. This settles the question of the existence of EF1 in the negative and therefore refutes the conjecture of \citet{DeligkasEKS2024}. Moreover, we show that the same lower bound holds for instances where the externalities are highly structured: even when all externalities are directed towards a single designated agent~$a_1$.

Before we conclude in \Cref{sec:conclusions} with open problems and future directions, in \Cref{sec:consensus}, we present an additional randomized mechanism~$\mathcal{M}_2$. This mechanism always produces an allocation that is EF up to $\Oh{n}$ items, which is a worse guarantee than that of mechanism~$\mathcal{M}_1$. However, unlike~$\mathcal{M}_1$, $\mathcal{M}_2$ produces a \emph{consensus} allocation, which means that we can arbitrarily permute the bundles while maintaining the fairness guarantee. Consequently, $\mathcal{M}_2$ is truthful in expectation, i.e., no agent can improve its expected utility by misreporting its valuations. This extends the result of~\citet{BuT2025} to the setting of externalities, with a slightly weaker bound on the relaxation of fairness that we obtain.


\subsection{Related Work}

Apart from the papers directly studying externalities in fair division of divisible and indivisible items mentioned above, there are several streams of related models in the literature.

First, there is a model of fair division with \emph{allocator's preferences}~\citep{BuLLST2023} or \emph{social impact}~\citep{FlamminiGV2025,DeligkasEGKS2026}. In both of these models, the goal is to find an allocation that is not only fair among agents, but also satisfies certain objectives of the allocator. This additional dimension of preferences can in principle be represented using externalities: the allocator is an additional agent with externalities towards all other agents. However, the crucial difference is that the allocator cannot receive any item and the original agents do not have any pairwise externalities. A more general model is \emph{fair division for groups}~\citep{ManurangsiS2017,Suksompong2018,SegalHaleviN2019,ManurangsiS2024,GolzY2025,DupreLaTour2026}, where agents are partitioned into groups and share items allocated to their group, while each agent retains their own subjective valuations, and for which constant relaxations of envy-freeness are known to not exist in most settings.

In \emph{asymptotic fair division}~\citep{ManurangsiS2020,ManurangsiS2021,ManurangsiS2025,ManurangsiSY2025,GargNS2026}, the task is not to decide which relaxation of fairness we can guarantee (since EF1 allocations always exist in the standard model), but rather to identify conditions under which stronger guarantees such as envy-freeness or proportionality can be achieved. Many of these works provide tight bounds with respect to the number of items required and rely heavily on results from discrepancy theory, similarly to our approach.

Externalities in preferences also appear in contexts beyond fair division. Notable examples include combinatorial auctions~\citep{Funk1996,JehielM1996,ZhangWWB2018}, matching~\citep{MumcuS2010,BranzeiMRLJ2013,FisherH2016,AnshelevichBH2017}, facility location~\citep{LiMXZZ2019,WangZL2024,LiPWZ2025}, and house allocation~\citep{MassandS2019,ElkindPTZ2020,AgarwalEGISV2021,GrossHumbertBBM2022,KnopS2023}. In these settings, an agent is typically assigned only one item (a partner, location, or house), and the solution concepts differ substantially from ours.


\section{Preliminaries}\label{sec:prelims}

In the fair division problem, the input is a set $M$ of indivisible items and a set $N$ of $n$ agents. The items in $M$ are to be allocated among the agents in $N$. Formally, an \emph{allocation} of the items is a tuple of pairwise disjoint subsets of $M$: $(A_1, \dots, A_n)$, where each $A_i \subseteq M$ is the bundle assigned to agent $i$. We say that allocation is \emph{complete} if $\bigcup_{i\in[n]} A_i = M$. Unless stated otherwise, we consider only complete allocations. Let $\mathcal{A}$ be the set of all allocations. Additionally, we denote the set $\{1,\ldots,n\}$ by $[n]$.

In the fair division with externalities model of~\citet{AzizSSW2023}, each agent $i \in N$ is associated with a function
\[
V_i \colon \mathcal{A} \to \mathbb{R},
\]
which assigns a value to every allocation $A \in \mathcal{A}$. We assume that valuations are \emph{additive}: for any agent
$i \in N$ and any allocation $(A_1,\dots,A_n) \in \mathcal{A}$,
\[
V_i(A) = \sum_{j \in N} \sum_{x \in A_j} V_i\bigl(x^{(j)}\bigr),
\]
where $x^{(j)} \in \mathcal{A}$ denotes the allocation in which agent $j$ receives
the single item $x$ and all other agents receive no items, that is,
\[
x^{(j)} = (\emptyset, \dots, \emptyset,
\underbrace{\{x\}}_{\mathclap{\text{$j^\text{th}$ coordinate}}},
\emptyset, \dots, \emptyset).
\]

To simplify the notation, we write $V_i(j,x)$ for $V_i\bigl(x^{(j)}\bigr)$ and therefore
\[
V_i(A) = \sum_{j \in N} \sum_{x \in A_j} V_i(j,x).
\]

Note that additivity implies, in particular, that valuations are
\emph{normalized}: for every agent $i \in N$, the empty allocation has value zero,
i.e., $V_i(\emptyset,\dots,\emptyset)=0$.

We say that the valuation $V_i$ of agent $i \in N$ has \emph{no chores} if, for every agent $j \in N$ and every item $x \in M$, we have
\[
V_i(i,x) \geq V_i(j,x).
\]
That is, agent $i$ weakly prefers receiving item $x$ herself to another agent receiving $x$.

We say that $V_i$ is \emph{binary} if, for all $j \in N$ and all $x \in M$,
\[
V_i(j,x) \in \{0,1\}.
\]


\subsection{Fairness in the Externalities Model}

Given an allocation $A \in \mathcal{A}$ and two distinct agents $i,j \in N$, we denote by $A^{i \leftrightarrow j}$ the allocation obtained from $A$ by swapping the bundles of agents $i$ and $j$, while the bundles of all other agents remain unchanged. In an allocation $A$, agent~$i$ \emph{envies} agent $j$ if
\[
V_i(A) < V_i\bigl(A^{i \leftrightarrow j}\bigr),
\]
that is, agent $i$ would prefer to swap her bundle with that of agent $j$. 
An allocation $A$ is \emph{envy-free} (EF) if no agent envies another agent. $A$ is \emph{envy-free up to $c$ items} (EF-$c$) if, for every pair of distinct agents $i,j \in N$, there exists a set of items $S \subseteq M$ with $|S| \leq c$ such that, defining $B = (A_1 \setminus S, \dots, A_n \setminus S)$, we have
\[
V_i(B) \geq V_i\bigl(B^{i \leftrightarrow j}\bigr).
\]
That is, for each pair of agents $i$ and $j$, there exists a set of at most $c$ items whose removal eliminates agent $i$’s envy toward agent $j$.

Note that, since valuations are assumed to be additive, the set $S$ used to eliminate agent $i$’s envy toward agent $j$ can always be chosen as a subset of $A_i \cup A_j$, as the items allocated to other agents contribute equally to $V_i(B)$ and $V_i(B^{i \leftrightarrow j})$. Moreover, if the valuation $V_i$ has no chores, the set $S$ can be chosen as a subset of $A_j$. Finally, in the absence of externalities (i.e. when $V_i(j,x) = 0$ for every item $x$ and pair of agents $i,j$) this definition coincides with the standard definition of envy-freeness up to $c$ items.


\section{Technical Tools}\label{sec:toolbox}


In this section, we present some technical tools that we use in our proofs. First, we present an alternate model of fair division instances, the asymmetric envy model, and prove that this model is equivalent for our purposes. Next, we describe the multicolor discrepancy problem and its connections to fair division. Most of our results derive fair division bounds from bounds in multicolor discrepancy. 


\subsection{Asymmetric Envy Model}

We now introduce the \emph{asymmetric envy model}, which is an alternate fair division model that allows for externalities among the agents' valuation functions. We will show that the asymmetric envy model is equivalent to the externalities model with respect to envy-based notions of fairness. We analyze most of our results using the asymmetric envy model, as the envy-based fairness notions in this model are easier to work with. This model does not assign a numerical value to each agent for a complete allocation; consequently, notions such as welfare maximization are not well defined. Instead, it specifies a notion of envy between pairs of distinct agents $i,j \in N$. The model is \emph{asymmetric} in the sense that the function defining the envy of agent $i$ toward agent $j$ depends on the ordered pair $(i,j)$. Envy from $i$ towards $j$ is captured by a valuation function $v_{i,j}$, and the standard notions of envy-freeness and envy-freeness up to $c$ items can be defined using these valuation functions. More precisely, for each pair of agents $i,j \in N$ with $i \neq j$, we are given a valuation function
\[
v_{i,j} \colon 2^M \to \mathbb{R}.
\]
We assume that these valuations are \emph{additive}, that is, for every set $S \subseteq M$,
\[
v_{i,j}(S) = \sum_{x \in S} v_{i,j}(x).
\]
We say that $v_{i,j}$ has \emph{no chores} if $v_{i,j}(x) \geq 0$ for all $x \in M$.
We say that the instance is \emph{binary} if $v_{i,j}(x)\in\{0,1\}$ for all $i\neq j$ and all $x\in M$.


\subsubsection{Fairness in the Asymmetric Envy Model}

Given an allocation $(A_1, \dots, A_n) \in \mathcal{A}$, agent $i$ is said to \emph{envy} agent $j$ if
\[
v_{i,j}(A_i) < v_{i,j}(A_j).
\]
An allocation $A$ is \emph{envy-free in the asymmetric envy model} (\EFa) if no agent envies another agent. $A$ is \emph{envy-free up to~$c$ items in the asymmetric envy model} (EF$_{a}$-$c$) if, for every pair of distinct agents $i,j \in N$, there exists a set of items $S \subseteq A_i \cup A_j$ with $|S| \leq c$ such that
\[
v_{i,j}(A_i \setminus S) \geq v_{i,j}(A_j \setminus S).
\]
Note that if $v_{i,j}$ has no chores, the set $S$ can always be chosen as a subset of $A_j$.

Next, we show that the asymmetric envy model is equivalent to the externalities model with respect to envy-based fairness notions such as envy-freeness and EF-$c$. That is, there exists a natural reduction from instances of the externalities model to those of the asymmetric envy model, such that any envy bound in the asymmetric envy model implies an equivalent bound in the externalities model.

\begin{proposition}\label{prop:equivalent}
    There exists a function $f$ that maps each instance $I$ of the externalities model to an instance $f(I)$ of the asymmetric envy model such that, for every allocation $A \in \mathcal{A}$, the allocation $A$ is EF (respectively, EF-$c$) in $I$ if and only if it is \EFa (respectively, \EFac) in $f(I)$. Moreover, if $I$ has no chores then $f(I)$ has no chores. Additionally: 
    \begin{itemize}
        \item For any instance $J$ in the asymmetric envy model with \emph{additive} valuations, there is an instance $I$ in the externalities model such that $J = f(I)$ (i.e., $f$ is surjective).
        \item For any instance $J$ in the asymmetric envy model with \emph{binary additive} valuations, there is a \emph{binary no-chores} instance $I$ in the externalities model such that $J = f(I)$.
    \end{itemize}
\end{proposition}
\begin{proof}
    Let $I$ be an externalities instance with valuations $(V_i)_{i\in N}$.
    Define $f(I)$ as follows: for every ordered pair $i\neq j$ and every item $x\in M$,
    \[
    v_{i,j}(x) \;:=\; V_i(i,x)-V_i(j,x),
    \]
    and extend additively to all bundles.
    
    Fix an allocation $A\in\mathcal{A}$ and distinct agents $i,j$.
    Since swapping $A_i$ and $A_j$ only changes the recipients of items in these two bundles, additivity gives
    \begin{align*}
    V_i\bigl(A^{i\leftrightarrow j}\bigr)-V_i(A)
    &=\sum_{x\in A_j}\bigl(V_i(i,x)-V_i(j,x)\bigr)\\
    &\;+\; \sum_{x\in A_i}\bigl(V_i(j,x)-V_i(i,x)\bigr)\\
    &=v_{i,j}(A_j)-v_{i,j}(A_i).
    \end{align*}
    Hence,
    \[
    V_i(A) < V_i\bigl(A^{i\leftrightarrow j}\bigr)
    \quad\Longleftrightarrow\quad
    v_{i,j}(A_i) < v_{i,j}(A_j),
    \]
    so envy comparisons coincide, and therefore EF in $I$ coincides with \EFa in $f(I)$.
    
    The same identity holds after discarding any set of items $S\subseteq M$ (i.e., for the allocation $A\setminus S$). Thus, for every pair $i\neq j$, there exists a set $S$ with $|S|\le c$ eliminating $i$'s envy toward $j$ in $I$ if and only if there exists such an $S$ eliminating envy in~$f(I)$. Therefore EF-$c$ in $I$ coincides with \EFac in~$f(I)$. Moreover, if $I$ has no chores, then for all $i\neq j$ and~$x$, $v_{i,j}(x)=V_i(i,x)-V_i(j,x)\ge 0$, thus $f(I)$ has no chores.
    
    For the additional claims:
    \begin{itemize}
        \item Given an additive instance $J$ with $(v_{i,j})_{i\neq j}$ in the asymmetric envy model, define $V_i(i,x)=1$ and $V_i(j,x)=1-v_{i,j}(x)$ for $j\neq i$. Then $V_i(i,x)-V_i(j,x)=v_{i,j}(x)$, so $f(I)=J$, implying the first claim.
        \item Given a binary additive instance $J$ with $(v_{i,j})_{i\neq j}$ in the asymmetric envy model, define $V_i(i,x)=1$ and $V_i(j,x)=1-v_{i,j}(x)$ for $j\neq i$. Then, since $v_{i,j}(x) \in \{0,1\}$, we have $V_i(j,x)\in\{0,1\}$, $V_i(i,x)\ge V_i(j,x)$, and $f(I)=J$, which imply the second claim.\qedhere
    \end{itemize}
\end{proof}


\subsection{Multicolor Discrepancy and Fair Division}

\citet{ManurangsiS2022} showed that tools from multicolor discrepancy can be leveraged to derive improved bounds for fairness notions of the form EF-$c$ in models with groups of agents. In this section, we describe the particular definitions and discrepancy guarantee that we will use later. We also extend their statement from the ``goods-only'' case to instances that may contain both goods and chores.

We present multicolor discrepancy in a formulation tailored to allocations of indivisible items. This is equivalent to the standard matrix-based definition, but is more convenient for our purposes. The version of multicolor discrepancy that we need is a recent generalization of~\citet{ManurangsiM2026}; it is an asymmetric setting in which valuations depend on the color.

Let $k \ge 1$ be an integer denoting the number of colors, and let $M$ be a set of $m$ items. Let $n_1 \ge \cdots \ge n_k$ be positive integers. For each $\ell \in [k]$, let $\mathcal{V}_\ell=\{v_{\ell,1},\dots,v_{\ell,n_\ell}\}$ be a collection of $n_\ell$ additive valuation functions, and define $\mathcal{V}=(\mathcal{V}_1,\dots,\mathcal{V}_k)$.

A \emph{$k$-coloring} of $M$ is a map $\chi:M\to [k]$, which induces a partition
\[
M=\chi^{-1}(1)\cup \cdots \cup \chi^{-1}(k).
\]

The \emph{asymmetric discrepancy} of $\mathcal{V}$ is defined as
\begin{align*}
\adisc(\mathcal{V},k)
:= &\min_{\chi:M\to [k]}\ \max_{\ell\in [k]}\\\ 
&\max_{i\in [n_\ell]}
\Bigl|\tfrac{1}{k}v_{\ell,i}(M)-v_{\ell,i}(\chi^{-1}(\ell))\Bigr|.
\end{align*}

In words, $\adisc(\mathcal{V},k)$ is the minimum, over all $k$-colorings of $M$, of the maximum (over colors and valuation functions) deviation between the value that $v_{\ell,i}$ assigns to its own color class and the proportional benchmark $\tfrac{1}{k}v_{\ell,i}(M)$.

\begin{theorem}[{\protect\cite[Theorem 5.4]{ManurangsiM2026}}]
    \label{thm:asym-disc}
    Assume that for every $\ell\in [k]$, $i\in [n_\ell]$, and $x\in M$, we have $v_{\ell,i}(x)\in [0,1]$. Then
    \[
    \adisc(\mathcal{V},k)\le \Oh{\sqrt{n_1}}.
    \]
    Moreover, a coloring achieving this bound can be found in deterministic polynomial time~\citep{LevyRR2017}.
\end{theorem}

In the symmetric case where $n = n_1 = \cdots = n_k$ and where $v_{\ell,i} = v_{\ell',i}$ for all $i \in [n]$ and all $\ell,\ell' \in [k]$, this result coincides exactly with Theorem~A.1 of~\citet{ManurangsiS2022}.

We will use Theorem~\ref{thm:asym-disc} by combining it with the following lemma to obtain \EFac guarantees. A related implication is implicit in the proof of the upper bound in~\citet{ManurangsiS2022}; here, we prove an extension that allows for a mixture of goods and chores.

\begin{restatable}{lemma}{lemdisctoefk}\label{lem:disc-to-efk}
    Let $k\ge 1$ and let $T$ be a positive integer threshold. Let $v:2^M\to \mathbb{R}$ be an additive valuation function that may assign both positive and negative values to items.
    Then there exist four auxiliary additive valuation functions $v^1,v^2,v^3,v^4:2^M\to \mathbb{R}$ such that $v^j(x)\in [0,1]$ for all $j\in[4]$ and all $x\in M$, and with the following property:
    If $A,B\subseteq M$ are two disjoint sets of items satisfying for all $j\in[4]$ and all $X\in\{A,B\}$
    \begin{equation}\label{eq:disc-to-efk-assump}
        \Bigl|\tfrac{1}{k}v^j(M)-v^j(X)\Bigr|\le T,
    \end{equation}
    then one can eliminate $A$'s envy for $B$ by discarding at most~$14T$ items from $A\cup B$, and similarly, one can eliminate $B$'s envy for $A$ by discarding at most $14T$ items from $A\cup B$.
\end{restatable}

\begin{proof}
Fix $k, T,$ and $v$. Let $M^+ := \{x \in M : v(x) > 0\}$ and $M^- := \{x \in M : v(x) < 0\}$. Let $L := \min\{m, 6Tk\}$, and let $S \subseteq M$ be any subset of $L$ items with the largest absolute values $|v(x)|$ (breaking ties arbitrarily). Let $S^+ := S \cap M^+$ and $S^- := S \cap M^-$. Let $R := M \setminus S$ denote the set of remaining items, and let $R^+ := R \cap M^+$ and $R^- := R \cap M^-$.

Define $p :=\min_{x\in S}|v(x)|$, then every ``small'' item $x\in R$ satisfies $|v(x)|\le p$. For each $x\in M$, set
\begin{align*}
 v^1(x)&=\mathbf{1}[x\in S^+],\\
 v^2(x)&=\mathbf{1}[x\in S^-],\\
 v^3(x)&=
\begin{cases}
v(x)/p, & x \in R^+, \text{ and } p>0,\\
0, & \text{otherwise},
\end{cases}\\
v^4(x)&=
\begin{cases}
-v(x)/p, & x\in R^- , \text{ and } p>0,\\
0, & \text{otherwise}.
\end{cases}
\end{align*}
Then $v^1,v^2,v^3,v^4$ are additive and satisfy $v^j(x)\in [0,1]$ for all $j$ and $x$.

For any $X\subseteq M$, we have $v^1(X)=|X\cap S^+|$ and $v^2(X)=|X\cap S^-|$, hence using~\eqref{eq:disc-to-efk-assump} (for $j=1,2$) yields, for each $X\in\{A,B\}$,
\begin{equation}
\label{eq:large-count-bounds}
\Bigl||X\cap S^+|-\tfrac{|S^+|}{k}\Bigr|\le T,
\qquad
\Bigl||X\cap S^-|-\tfrac{|S^-|}{k}\Bigr|\le T.
\end{equation}
In particular, since $|S^+| + |S^-| = L\le 6Tk$,
\[
|B\cap S^+|\le \tfrac{L}{k}+T\le 7T,
\qquad
|A\cap S^-|\le \tfrac{L}{k}+T\le 7T.
\]

Let $P:=B\cap S^+$ and $Q:=A\cap S^-$. By the previous inequality, $|P|,|Q|\le 7T$. We will show that
\begin{equation}
\label{eq:main-ineq}
v(A\setminus Q)\ge v(B\setminus P).
\end{equation}
This implies that any envy of $A$ toward $B$ can be eliminated by discarding at most $|P|+|Q|\le 14T$ items.
By symmetry (swapping $A$ and $B$), the same holds in the other direction, so the lemma follows.

\medskip
\noindent\emph{Case 1: $L=m$ (so $S=M$).}
Then $R=\emptyset$. By construction, $B\setminus P$ contains no goods and $A\setminus Q$ contains no chores,
so $v(A\setminus Q)\ge 0\ge v(B\setminus P)$, proving~\eqref{eq:main-ineq}.

\medskip
\noindent\emph{Case 2: $L=6Tk$ (so $m\ge 6Tk$).}
If $p=0$, then $v(x)=0$ for all $x\in R$, hence $v(A\cap R)=v(B\cap R)=0$ and the argument from Case~1 applies.

Assume $p>0$. Using~\eqref{eq:disc-to-efk-assump} for $j=3,4$ and multiplying by $p$, for $X\in\{A,B\}$ we have
\begin{align*}
    \Bigl|\tfrac{1}{k}v(R^+)-v(X\cap R^+)\Bigr|\le pT,\\
    \Bigl|\tfrac{1}{k}v(R^-)-v(X\cap R^-)\Bigr|\le pT.
\end{align*}
Hence
\[
|v(A\cap R^+)-v(B\cap R^+)|\le 2pT,\]
\[|v(A\cap R^-)-v(B\cap R^-)|\le 2pT,
\]
and therefore
\begin{equation}
\label{eq:small-diff}
v(A\cap R)-v(B\cap R)\ge -4pT.
\end{equation}

Since every $x\in S$ satisfies $|v(x)|\ge p$, we have $v(x)\ge p$ for $x\in S^+$ and $v(x)\le -p$ for $x\in S^-$.
Moreover, $A\setminus Q$ contains $A\cap S^+$ and no items from $A\cap S^-$, while $B\setminus P$ contains $B\cap S^-$ and no
items from $B\cap S^+$. Hence
\begin{align}
\label{eq:split-large-small}
v(A\setminus Q)-v(B\setminus P)
&= \bigl(v(A\cap S^+)-v(B\cap S^-)\bigr)  \\
&\quad + \bigl(v(A\cap R)-v(B\cap R)\bigr). \notag
\end{align}
We lower bound the first term in~\eqref{eq:split-large-small} as
\begin{align}
v(A\cap S^+)-v(B\cap S^-)
&\ge p\bigl(|A\cap S^+|+|B\cap S^-|\bigr).\label{eq:large-lb}
\end{align}
By~\eqref{eq:large-count-bounds},
\[
|A\cap S^+|\ge \tfrac{|S^+|}{k}-T,
\qquad
|B\cap S^-|\ge \tfrac{|S^-|}{k}-T,
\]
so
\begin{equation}
\label{eq:count-sum}
|A\cap S^+|+|B\cap S^-|
\ge \tfrac{|S^+|+|S^-|}{k}-2T
= \tfrac{L}{k}-2T
= 4T.
\end{equation}
Combining~\eqref{eq:large-lb} and~\eqref{eq:count-sum} yields
\begin{equation}
\label{eq:large-diff}
v(A\cap S^+)-v(B\cap S^-)\ge 4pT.
\end{equation}

Finally, plugging~\eqref{eq:small-diff} and~\eqref{eq:large-diff} into~\eqref{eq:split-large-small} gives
\[
v(A\setminus Q)-v(B\setminus P)\ge 4pT-4pT=0,
\]
which proves~\eqref{eq:main-ineq}.
\end{proof}

For the lower bound, we will use the closely related notion of \emph{$p$-weighted discrepancy}. Given a set $V=\{v_1,\dots,v_n\}$ of~$n$ additive valuation functions and a parameter $p\in(0,1)$, the $p$-weighted discrepancy is defined as
\begin{equation*}
\textstyle \wdisc_{p}(V)
:= \displaystyle \min_{A \subseteq M} \ \max_{i\in [n]}
\Bigl|p\cdot v_i(M)-v_i(A)\Bigr|.
\end{equation*}

Theorem 3.2 of \citet{ManurangsiM2026} establishes the following lower bound.

\begin{theorem}[\protect\cite{ManurangsiM2026}]\label{thm:manurangsimekaLB}
    For every $p\in(0,1)$ and $n\in\mathbb{N}$, there exist a set $M$ of items and a collection $V=\{v_1,\dots,v_n\}$ of $n$ binary additive valuation functions such that
    \[ \textstyle
    \wdisc_p(V)\ge \frac{\sqrt{n-1}}{16}.
    \]
\end{theorem}


\section{An $\Oh{\sqrt{n}}$ Upper Bound}\label{sec:upperbound}

We now combine Lemma~\ref{lem:disc-to-efk} with the asymmetric discrepancy result of~\citet{ManurangsiM2026} (Theorem~\ref{thm:asym-disc}) to obtain an $\Oh{\sqrt{n}}$ bound in the asymmetric envy model. This construction is \emph{non-consensus}, i.e. it produces an allocation $(A_1,\dots,A_n)$ tailored to the agent labels, and the guarantee does not necessarily hold under permutations of the bundles.

\begin{theorem}\label{thm:asym-envy-nonconsensus}
    Let $M$ be a set of items and let $(v_{i,j})_{i,j\in N,\,i\neq j}$ be an instance of the asymmetric envy model with $n$ agents, where each $v_{i,j}$ is additive and may assign both positive and negative values to items. Then there exists a partition
    \[
    M = A_1 \sqcup \cdots \sqcup A_n
    \]
    such that for every ordered pair of distinct agents $i,j\in N$, one can eliminate $i$'s envy toward $j$ (measured by $v_{i,j}$) between $A_i$ and $A_j$ by discarding at most $\Oh{\sqrt{n}}$ items from $A_i\cup A_j$. Equivalently, the allocation $A=(A_1,\dots,A_n)$ is $\text{EF}_a$-$\Oh{\sqrt{n}}$. Moreover, such an allocation can be computed in deterministic polynomial time.
\end{theorem}
\begin{proof}
    Let $k:=n$ (the number of colors). Let $T=\Oh{\sqrt{n}}$ be the discrepancy bound promised by Theorem~\ref{thm:asym-disc} when applied with
    \[
    n_1 = 8(n-1).
    \]
    
    For every ordered pair $i\neq j$, apply Lemma~\ref{lem:disc-to-efk} to the valuation function $v_{i,j}$ with parameters $k=n$ and $T$. This yields four additive valuation functions
    \[
    v_{i,j}^1,\ v_{i,j}^2,\ v_{i,j}^3,\ v_{i,j}^4:2^M\to\mathbb{R},
    \]
    with $v_{i,j}^t(x)\in[0,1]$ for all $t\in\{1,2,3,4\}$ and $x\in M$.
    
    We now build an asymmetric discrepancy instance with $k=n$ colors as follows.
    For each color $\ell\in[n]$, define the multiset of valuation functions
    \begin{align*}
        \mathcal{V}_\ell
    \;:=\;
    &\{\, v_{\ell,j}^t : j\in[n]\setminus\{\ell\},\ t\in\{1,2,3,4\}\,\} \ \cup\ \\
    &\{\, v_{j,\ell}^t : j\in[n]\setminus\{\ell\},\ t\in\{1,2,3,4\}\,\}.
    \end{align*}

    Thus $|\mathcal{V}_\ell| = 8(n-1)$ for every $\ell$, and all functions in all~$\mathcal{V}_\ell$ satisfy the boundedness condition of Theorem~\ref{thm:asym-disc}.
    
    Applying Theorem~\ref{thm:asym-disc} yields a coloring $\chi:M\to[n]$ such that for every $\ell\in[n]$ and every $u\in\mathcal{V}_\ell$,
    \[
    \Bigl|\tfrac{1}{n}u(M)-u(\chi^{-1}(\ell))\Bigr|\le T.
    \]
    Let $A_\ell:=\chi^{-1}(\ell)$. Then $(A_1,\dots,A_n)$ is a partition of $M$.
    
    Fix any ordered pair $i\neq j$. By construction, the four functions $v_{i,j}^1,\dots,v_{i,j}^4$ belong to both $\mathcal{V}_i$ and $\mathcal{V}_j$, so for each $t\in\{1,2,3,4\}$ we have
    \[
    \Bigl|\tfrac{1}{n}v_{i,j}^t(M)-v_{i,j}^t(A_i)\Bigr|\le T
    \]
    and
    \[
    \Bigl|\tfrac{1}{n}v_{i,j}^t(M)-v_{i,j}^t(A_j)\Bigr|\le T.
    \]
    Thus condition~\eqref{eq:disc-to-efk-assump} of Lemma~\ref{lem:disc-to-efk} holds for the disjoint sets $A_i$ and $A_j$ (with $k=n$), and Lemma~\ref{lem:disc-to-efk} implies that $i$'s envy toward $j$ with respect to $v_{i,j}$ can be eliminated by discarding at most $14T=\Oh{\sqrt{n}}$ items from $A_i\cup A_j$.
    
    Since this holds for every ordered pair $i\neq j$, the allocation is $\text{EF}_a$-$\Oh{\sqrt{n}}$.
    Deterministic polynomial-time computability follows from Theorem~\ref{thm:asym-disc}.
\end{proof}

By combining \Cref{thm:asym-envy-nonconsensus} with repeated application of \Cref{prop:equivalent}, we obtain the same upper-bound also for the externalities model.

\begin{restatable}{corollary}{corextnoncon}\label{cor:externalities-nonconsensus}
    Given an instance of the externalities model with $n$ agents, there exists a deterministic polynomial-time algorithm that outputs an allocation $(A_1,\dots,A_n)$ that is \mbox{$\text{EF}$-$\Oh{\sqrt{n}}$}.
\end{restatable}
\begin{proof}
    Given an externalities instance $I$, use Proposition~\ref{prop:equivalent} to obtain an equivalent asymmetric envy instance $f(I)$ with valuations $v_{i,j}(x)=V_i(i,x)-V_i(j,x)$. Then, apply Theorem~\ref{thm:asym-envy-nonconsensus} to $f(I)$ to obtain an $\text{EF}_a$-$\Oh{\sqrt{n}}$ allocation. By Proposition~\ref{prop:equivalent}, the same allocation is EF$-\Oh{\sqrt{n}}$ in the externalities model.
\end{proof}


\section{An Asymptotically Tight Lower Bound}\label{sec:lowerbound}

In this section, we prove a matching lower bound, showing that Corollary~\ref{cor:externalities-nonconsensus} is asymptotically tight, even if the instance is binary additive and with no chores. For simplicity, we first prove our lower bound in the asymmetric envy model with binary valuations. We then use Proposition~\ref{prop:equivalent} to show that this result implies an equivalent lower bound in the externalities model with binary no-chores valuations.

Let $n$ be any positive odd integer, $q=(n-1)/2$, and $S_1, \ldots, S_q$ be $q$ subsets of a set of items $M$. For any $i$, let $v^{S_i}$ be the binary additive valuation function defined by $v^{S_i}(A):= |A\cap S_i|$ for all $A \subseteq M$. Let $\mathcal{V} = \{v^{S_1},\dots, v^{S_q}\}$.

Using the above valuation profile, we can construct an associated instance of the asymmetric envy model as follows.

For all $j\geq 2$, we let $v_{1,j} (A):= |A|$.
For each $i \in [q]$, we define $\bar S_i:=M\setminus S_i$. Then, for each $i \in [q]$, for all $j\in[n]\setminus\{2i\}$, we let
\[
v_{2i,j}(A):=
\begin{cases}
v^{S_i}(A), & \text{if } j=1,\\
v^{\bar S_i}(A), & \text{otherwise},
\end{cases}
\]
and for all $j\in[n]\setminus\{2i+1\}$,
\[
v_{2i+1,j}(A):=
\begin{cases}
v^{\bar S_i}(A), & \text{if } j=1,\\
v^{S_i}(A), & \text{otherwise}.
\end{cases}
\]

We show that if this instance admits an \EFac allocation for some integer $c$, then its $1/n$-weighted discrepancy is at most~$6c$.

\begin{lemma}\label{lem:lowerbound}
    In the above instance, suppose there exists a partition $M=A_1\sqcup\cdots\sqcup A_n$ that is \EFac for some integer $c\ge 1$. Then
    \[
    \textstyle \wdisc_{1/n}(\mathcal{V})\le 6c.
    \]
\end{lemma}
\begin{proof}
    Fix an \EFac allocation $M=A_1\sqcup\cdots\sqcup A_n$ and write $m=|M|$.
    By the definition of $\wdisc_{1/n}(\mathcal{V})$, choosing the set $A_1\subseteq M$ gives
    \[
    {\textstyle\wdisc_{1/n}(\mathcal{V})}
    \le
    \max_{i\in[q]}
    \left|
    \frac1n\,v^{S_i}(M)-v^{S_i}(A_1)
    \right|.
    \]
    Since $v^{S_i}(X)=|S_i\cap X|$, we have $v^{S_i}(M)=|S_i|$ and $v^{S_i}(A_1)=|S_i\cap A_1|$, hence it suffices to show that for every~$i\in[q]$,
    \[
    \left|
    |S_i\cap A_1|-\frac{|S_i|}{n}
    \right|
    \le 6c.
    \]
    Fix $i$ and abbreviate $\bar S_i=M\setminus S_i$.
    
    \medskip
    \noindent\textbf{Step 1: all bundles have almost the same size.}
    Recall that agent $1$ values every item at $1$, so for every bundle $X$ and every $j\ge 2$ we have $v_{1,j}(X)=|X|$.
    Since the allocation is \EFac, for every $j\ge 2$,
    \[
    |A_1|=v_{1,j}(A_1)\ge v_{1,j}(A_j)-c=|A_j|-c.
    \]
    Moreover, by the definition of the valuation profile (for this fixed $i$),
    agent $2i$ counts $S_i$-items when comparing to agent~$1$ and counts $\bar S_i$-items when comparing to agent $2i+1$,
    and agent $2i+1$ behaves symmetrically.
    Applying \EFac along these comparisons yields
    \[
    |S_i\cap A_{2i}|\ge |S_i\cap A_1|-c,
    \quad
    |\bar S_i\cap A_{2i}|\ge |\bar S_i\cap A_{2i+1}|-c,
    \]
    and
    \[
    |S_i\cap A_{2i+1}|\ge |S_i\cap A_{2i}|-c,
    \quad
    |\bar S_i\cap A_{2i+1}|\ge |\bar S_i\cap A_1|-c.
    \]
    Combining,
    \[
    \begin{aligned}
    |A_{2i}|
    &=|S_i\cap A_{2i}|+|\bar S_i\cap A_{2i}| \\
    &\ge (|S_i\cap A_1|-c)+(|\bar S_i\cap A_{2i+1}|-c) \\
    &\ge |S_i\cap A_1|-c+(|\bar S_i\cap A_1|-2c) \\
    &=|A_1|-3c.
    \end{aligned}
    \]
    The same argument gives $|A_{2i+1}|\ge |A_1|-3c$.
    
    Summing $\sum_{j=1}^n |A_j|=m$ and using these inequalities gives
    \[
    m\ge |A_1|+(n-1)(|A_1|-3c)=n|A_1|-3c(n-1),
    \]
    hence
    \[
    |A_1|\le \frac{m}{n}+3c.
    \]
    
    \medskip
    \noindent\textbf{Step 2: discrepancy for $S_i$.}
    We distinguish two cases.
    
    \smallskip
    \noindent\emph{Case 1: $|S_i\cap A_1|\le |S_i|/n$.}
    Since $S_i=\bigsqcup_{j=1}^n (S_i\cap A_j)$, by averaging there exists some $j\ne 1$ with
    $|S_i\cap A_j|\ge |S_i|/n$.
    If $j\ne 2i+1$, then \EFac for agent $2i+1$ gives
    $|S_i\cap A_{2i+1}|\ge |S_i\cap A_j|-c\ge |S_i|/n-c$; if $j=2i+1$ this is immediate.
    Thus
    \[
    |S_i\cap A_{2i+1}|\ge \frac{|S_i|}{n}-c.
    \]
    Together with $|\bar S_i\cap A_{2i+1}|\ge |\bar S_i\cap A_1|-c$ we obtain
    \[
    |A_{2i+1}|
    =|S_i\cap A_{2i+1}|+|\bar S_i\cap A_{2i+1}|
    \ge \frac{|S_i|}{n}+|\bar S_i\cap A_1|-2c.
    \]
    Using $|A_1|\ge |A_{2i+1}|-c$ yields
    \[
    |A_1|\ge \frac{|S_i|}{n}+|\bar S_i\cap A_1|-3c.
    \]
    Since $|A_1|=|S_i\cap A_1|+|\bar S_i\cap A_1|$, we conclude
    \[
    |S_i\cap A_1|\ge \frac{|S_i|}{n}-3c,
    \]
    and therefore in Case~1,
    \[
    \left||S_i\cap A_1|-\frac{|S_i|}{n}\right|\le 3c.
    \]
    
    \smallskip
    \noindent\emph{Case 2: $|S_i\cap A_1|\ge |S_i|/n$.}
    We first claim that
    \[
    |\bar S_i\cap A_1|\ge \frac{|\bar S_i|}{n}-3c.
    \]
    Indeed, if $|\bar S_i\cap A_1|\ge |\bar S_i|/n$ there is nothing to prove.
    Otherwise, by averaging there exists some $j\ne 1$ with $|\bar S_i\cap A_j|\ge |\bar S_i|/n$,
    and repeating the argument of Case~1 with $\bar S_i$ in place of $S_i$ (using agent $2i$ instead of $2i+1$)
    gives the bound.
    
    Using the claim and the bound $|A_1|\le m/n+3c$ from Step~1,
    \[
    \begin{aligned}
    |S_i\cap A_1|
    &=|A_1|-|\bar S_i\cap A_1| \\
    &\le \left(\frac{m}{n}+3c\right)-\left(\frac{|\bar S_i|}{n}-3c\right)
    = \frac{|S_i|}{n}+6c.
    \end{aligned}
    \]
    Since Case~2 assumes $|S_i\cap A_1|\ge |S_i|/n$, we get
    \[
    \left||S_i\cap A_1|-\frac{|S_i|}{n}\right|\le 6c.
    \]
    
    \medskip
    In both cases we have shown $\bigl||S_i\cap A_1|-|S_i|/n\bigr|\le 6c$.
    Taking the maximum over $i\in[q]$ and recalling $v^{S_i}(X)=|S_i\cap X|$ yields
    \[
    \max_{i\in[q]}
    \left|
    \frac1n\,v^{S_i}(M)-v^{S_i}(A_1)
    \right|
    \le 6c,
    \]
    and therefore $\wdisc_{1/n}(\mathcal{V})\le 6c$.
\end{proof}

We now formally present our lower bound for the case of binary no-chores valuations in the externalities model.

\begin{theorem}\label{thm:envy-lb-binarync}
    There is a binary no-chores instance of the externalities model with $n$ agents for which every EF-$c$ allocation requires $c=\Omh{\sqrt{n}}$.
\end{theorem}
\begin{proof}
    The proof of this theorem follows by combining the above lemma with Theorem~\ref{thm:manurangsimekaLB}. The instance of \citet{ManurangsiM2026} consists of $\Theta(n)$ horizontally-stacked copies of $\frac{1}{2}(\boldsymbol{1}_{q\times q}+\boldsymbol{H})$, where $\boldsymbol{1}_{q\times q}$ is a $q\times q$ all-1s matrix and $\boldsymbol{H}$ is a $q\times q$ Hadamard matrix. Clearly, this instance is binary. If we let the sets $S_1, \dots, S_{q}$ be given by the hyperedges in this instance, then the corresponding valuation profile $\mathcal{V}$ consists of $q$ binary additive functions. Taking $p=1/n$, Theorem~\ref{thm:manurangsimekaLB} gives us the following bound:
    \[
        {\textstyle\wdisc_{1/n}(\mathcal{V})}\ge \frac{\sqrt{q-1}}{16}.
    \]
    Since $q=(n-1)/2$, combining this with Lemma~\ref{lem:lowerbound} we get 
    \[
        c\ge\frac{\sqrt{((n-1)/2)-1}}{96},
    \]
    and thus $c=\Omh{\sqrt{n}}$ and the corresponding instance in the asymmetric envy model admits no \EFac allocation for any $c \in o(\sqrt{n})$. Finally, by Proposition~\ref{prop:equivalent}, there exists a binary no-chores instance in the externalities model for which every EF$-c$ allocation requires $c=\Omh{\sqrt{n}}$.
\end{proof}

As stated previously, an important implication of this result is that EF1 allocations do not always exist when the agents have externalities, even when the valuations are binary no-chores, which refutes a conjecture of \citet{DeligkasEKS2024} and resolves the open question of \citet{AzizSSW2023}.

\begin{theorem}[Corollary of Theorem~\ref{thm:envy-lb-binarync}]\label{thm:binary-counterexample}
    There is a binary no-chores instance of the externalities model that admits no EF1 allocation.
\end{theorem}

Interestingly, the above example in the asymmetric envy model can also be used to construct an instance in the externalities model that is not binary, but instead has the property that agent~1 has no externalities and every other agent in the instance only has externalities towards agent~1. Together, Theorems~\ref{thm:envy-lb-binarync} and~\ref{thm:envy-lb-star} show that even highly structured instances with externalities do not allow for EF1 allocations.

\begin{restatable}{theorem}{thmenvylbstar}\label{thm:envy-lb-star}
    There is an instance in the externalities model with $n$ agents in which every agent has externalities towards only agent 1, and for which every EF-$c$ allocation requires $c=\Omh{\sqrt{n}}$.
\end{restatable}

\begin{proof}
Consider the instance $(v_{i,j})_{i,j\in[n]}$ from Lemma~\ref{lem:lowerbound} in the asymmetric envy model. We let $V_1(1,x)=1$ and $V_1(j,x)=0$ for all $x\in M$ and $j\ge2$. For each $i \in [q]$, if $x\in S_i$, we let
\[
V_{2i}(j,x):=
\begin{cases}
-1, & \text{if } j=1,\\
0, & \text{otherwise},
\end{cases}
\]
and otherwise, if $x\in\bar S_i$, we let
\[
V_{2i}(j,x):=
\begin{cases}
1, & \text{if } j=1,\\
1, & \text{if } j=2i,\\
0, & \text{otherwise}.
\end{cases}
\]

Similarly, for each $i \in [q]$, if $x\in\bar S_i$, we let
\[
V_{2i+1}(j,x):=
\begin{cases}
-1, & \text{if } j=1,\\
0, & \text{otherwise},
\end{cases}
\]
and otherwise, if $x\in S_i$, we let
\[
V_{2i+1}(j,x):=
\begin{cases}
1, & \text{if } j=1,\\
1, & \text{if } j=2i+1,\\
0, & \text{otherwise}.
\end{cases}
\]

The above valuation profile in the externalities model satisfies
\[
v_{i,j}(x) \;:=\; V_i(i,x)-V_i(j,x),
\]
for all items $x\in M$ and all agents $i,j\in[n]$ in the asymmetric envy model instance of Lemma~\ref{lem:lowerbound}, and we can thus apply the bound. Clearly, no agent has externalities towards any other agent except for agent~1. By Proposition~\ref{prop:equivalent}, Lemma~\ref{lem:lowerbound} and Theorem~\ref{thm:envy-lb-binarync}, every EF-$c$ allocation in this instance requires $c=\Omh{\sqrt{n}}$.
\end{proof}


\section{An Ex-Ante Truthful Mechanism}\label{sec:consensus}

In this section, we present a \emph{consensus} version of our upper bound, in which the bundles of the resulting allocation can be permuted arbitrarily amongst the agents while maintaining the fairness guarantee. This allows us to implement the resulting construction as a mechanism that is \textit{truthful in expectation}, i.e., in which no agent can improve its expected value (ex-ante) by misreporting its valuation. Here, the expectation is over the randomness of the mechanism, and truthfulness follows directly from the fact that the agents can be matched to the bundles uniformly at random.

Specifically, we combine \Cref{lem:disc-to-efk} and a symmetric application of \Cref{thm:asym-disc} to obtain a generalization of the consensus upper bound of \citet[Theorem 1.3]{ManurangsiS2022} to a mixture of goods and chores.

\begin{restatable}{theorem}{thmconsensus}\label{thm:consensus}
    Let $v_1,\dots,v_n$ be additive valuation functions that may assign both positive and negative values to items, and let $k$ be a number of colors. There exists a partition of the items into $k$ bundles $M = A_1\sqcup \dots\sqcup A_k$ such that for any $\ell,\ell' \in [k]$ and any $i \in [n]$, one can eliminate envy (with respect to $v_i$) between $A_\ell$ and $A_{\ell'}$ by discarding at most $\Oh{\sqrt{n}}$ items. Moreover, such a partition can be computed in deterministic polynomial time.
\end{restatable}

\begin{proof}
Let $T = \Oh{\sqrt{n}}$ denote the discrepancy bound guaranteed by
Theorem~\ref{thm:asym-disc} when applied with $n_1 = 4n$.

We first apply Lemma~\ref{lem:disc-to-efk} to each valuation function
$v_i$, for $i \in [n]$, with parameters $k$ and $T$. This yields,
for every $i$, four additive valuation functions
\[
v_i^1, v_i^2, v_i^3, v_i^4 : 2^M \to \mathbb{R},
\]
such that $v_i^j(x) \in [0,1]$ for all $x \in M$ and all
$j \in [4]$.

Next, we construct a symmetric discrepancy instance as follows. For each
color $\ell \in [k]$, define the multiset of valuation functions
\[
\mathcal{V}_\ell
\;=\;
\{v_1^1, v_1^2, v_1^3, v_1^4, \dots, v_n^1, v_n^2, v_n^3, v_n^4\}.
\]
Thus, for every $\ell$, we have $n_\ell = 4n$, and all valuation
functions satisfy the boundedness condition required by
Theorem~\ref{thm:asym-disc}.

Applying Theorem~\ref{thm:asym-disc} to this instance yields a
$k$-coloring of the items, which induces a partition
\[
M = A_1 \sqcup \cdots \sqcup A_k,
\]
such that for every $\ell \in [k]$, every $i \in [n]$, and every
$j \in [4]$,
\[
\Bigl| \tfrac{1}{k} v_i^j(M) - v_i^j(A_\ell) \Bigr| \;\le\; T.
\]

Fix any two bundles $A_\ell$ and $A_{\ell'}$ and any index $i\in[n]$.
The above inequalities verify condition~\eqref{eq:disc-to-efk-assump} of
Lemma~\ref{lem:disc-to-efk} for the sets $A_\ell$ and $A_{\ell'}$.
Hence, Lemma~\ref{lem:disc-to-efk} implies that the envy measured with
respect to the original valuation $v_i$ between $A_\ell$ and $A_{\ell'}$
can be eliminated by discarding at most $14T = \Oh{\sqrt{n}}$ items.
\end{proof}

The corollary below follows by combining an application of Theorem~\ref{thm:consensus} with the simple mechanism described in~\citet{BuT2025} that assigns the bundles uniformly at random.

\begin{restatable}{corollary}{cortruthful}\label{cor:truthful}
    Given $n$ agents with additive valuation functions that may assign both positive and negative values to items, there exists a randomized ex-ante truthful mechanism running in polynomial time that outputs an allocation $(A_1,\dots,A_n)$ that is EF-$\Oh{\sqrt{n}}$.
\end{restatable}
\begin{proof}
The corollary follows from using Theorem~\ref{thm:consensus} with $k = n$ colors, together with the simple mechanism described in~\citet{BuT2025}, which assigns the bundles uniformly at random. The expected value received by agent $i$ is $\tfrac{1}{n} v_i(M)$, regardless of the reports of the agents. In fact, this guarantee is much stronger than truthfulness: no coalition of agents can influence an agent’s expected value.
\end{proof}

In the externalities model, we get the following corollary by applying Theorem~\ref{thm:consensus} to the $n(n-1)$ asymmetric envy valuations $(v_{i,j})_{i\neq j}$ produced by Proposition~\ref{prop:equivalent}. This yields a consensus bound of order $\Oh{\sqrt{n(n-1)}}=\Oh{n}$.

\begin{restatable}{corollary}{corextcon}\label{cor:externalities-consensus}
    Given an instance of the externalities model, there exists a randomized ex-ante truthful mechanism running in polynomial time that outputs an allocation $(A_1,\dots,A_n)$ that is EF-$\Oh{n}$.
\end{restatable}

\begin{proof}
Given an instance $I$ of the externalities model, use Proposition~\ref{prop:equivalent} to obtain an equivalent instance $f(I)$ in the asymmetric envy model with respect to envy-based comparisons. Apply Theorem~\ref{thm:consensus} to the collection of $n(n-1)$ valuation functions $\{v_{i,j} : i,j\in N,\, i\neq j\}$ with $k=n$ colors. This produces a partition $M=B_1\sqcup\cdots\sqcup B_n$ such that every $v_{i,j}$ is ``balanced'' between any two bundles up to discarding $\Oh{\sqrt{n(n-1)}}=\Oh{n}$ items.

Finally, assign the bundles uniformly at random to the agents as in~\citet{BuT2025}. Since each item is assigned to each agent with probability $1/n$, the expected value of agent $i$ is
\[
\frac{1}{n}\sum_{x\in M}\sum_{j\in N} V_i(j,x),
\]
regardless of the reports of the agents. Because the partition is consensus (i.e., it works for every pair of bundles), the resulting allocation is EF-$\Oh{n}$ in $f(I)$, and therefore also EF-$\Oh{n}$ in $I$ by Proposition~\ref{prop:equivalent}.
\end{proof}


\section{Conclusions}\label{sec:conclusions}

For the fair division of indivisible items amongst agents with externalities, we find asymptotically tight bounds on the optimal relaxation EF-$k$ that can always be attained. Specifically, for the general case of additive valuations, we show that an EF-$\Oh{\sqrt{n}}$ allocation always exists, and give a matching $\Omh{\sqrt{n}}$ lower bound even for the special case of binary valuations with no chores. In the process, we answer two open questions of \citet{AzizSSW2023}, and resolve a conjecture of \citet{DeligkasEKS2024}, showing that EF1 allocations do not always exist when the agents have externalities. We also present an ex-ante-truthful consensus-based mechanism that finds an EF-$\Oh{n}$ allocation when the agents have externalities. Our work raises interesting new directions for research.

\begin{itemize}
    \item \emph{Small explicit counterexample}: Our proof technique shows how to reduce a known lower bound from discrepancy into a collection of lower bound instances for the externalities model. Nonetheless, we leave open the question of describing a small explicit counterexample, and of determining the largest value of $n$ for which an EF1 allocation exists for additive/binary valuations with or without the no-chores condition.
    
    \item \emph{Consensus lower-bound}: We give an ex-ante-truthful consensus-based mechanism that finds an EF-$\Oh{n}$ allocation when the agents have externalities. Does there exist a matching lower bound in this setting?
    
    \item \emph{Other settings}: Can our approach be applied to find similar bounds in other settings, such as fair division with social impact~\citep{FlamminiGV2025,DeligkasEGKS2026}?
\end{itemize}


\appendix
\clearpage
\section*{Acknowledgments}
This project has received funding from the European Research Council (ERC) under the European Union’s Horizon 2020 research and innovation programme (grant agreement No 101002854) and by the European Union under the project Robotics and advanced industrial production (reg. no. CZ.02.01.01/00/22\_008/0004590). 
The first author has been supported by NSERC Discovery Grant 2022-04191. We are grateful to Adrian Vetta for discussions.
    
\begin{center}
    \includegraphics[width=4cm]{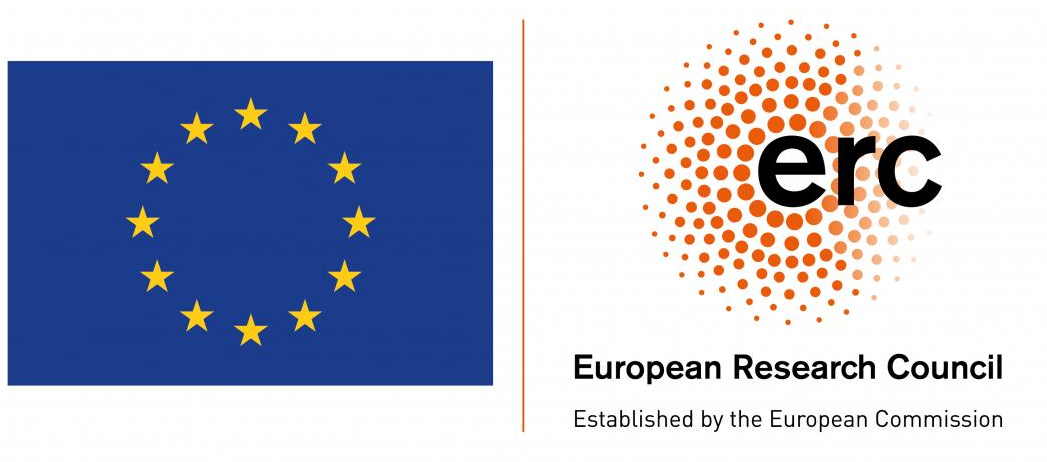}
\end{center}


\bibliographystyle{named}
\bibliography{references}

\end{document}